\documentclass{llncs}

\usepackage[utf8]{inputenc}

\newcommand{\proc}[1]{\textnormal{\scshape#1}}
\newcommand{\Mod}[1]{\ \text{mod}\ #1}
\newcommand{\rsec}[1]{\mbox{Section \ref{#1}}}
\newcommand{\rfig}[1]{\mbox{Figure \ref{fig:#1}}}
\usepackage{amssymb,amsmath,algorithm,enumerate,color,xspace}
\usepackage[disable]{todonotes}
\usepackage{graphicx}
\usepackage{mathtools,paralist}

\newcommand{\rthm}[1]{Theorem \ref{thm:#1}}
\newcommand{\rlem}[1]{Lemma \ref{lem:#1}}

\newcommand{\fig}[6]{
  \begin{figure}[#1] 
    \centering
    \begin{minipage}{#2\linewidth}
      \centering
      {\includegraphics[#3]{./figures/#4}}
      \caption{#6}
      \label{fig:#5}
    \end{minipage}
  \end{figure}
}

\title{Strictly Implicit Priority Queues: \\ On the Number of Moves and Worst-Case Time}

\author{Gerth Stølting Brodal 
  \and Jesper Sindahl Nielsen 
  \and Jakob Truelsen}

\institute{MADALGO\thanks{Work supported in part by the Danish
    National Research Foundation grant DNRF84 through the Center for
    Massive Data Algorithmics.}, Department of Computer Science, Aarhus
  University, Denmark \email{$\{$gerth,jasn,jakobt$\}$@cs.au.dk} }

\begin{document}

\maketitle

\begin{abstract}
   The binary heap of Williams (1964) is a simple priority queue
   characterized by only storing an array containing the elements and
   the number of elements~$n$ -- here denoted a \emph{strictly
     implicit} priority queue.  We introduce two new strictly implicit
   priority queues.  The first structure supports amortized $O(1)$
   time \proc{Insert} and $O(\log n)$ time \proc{ExtractMin}
   operations, where both operations require amortized $O(1)$ element
   moves.  No previous implicit heap with $O(1)$ time \proc{Insert}
   supports both operations with $O(1)$ moves.  The second structure
   supports worst-case $O(1)$ time \proc{Insert} and $O(\log n)$ time
   (and moves) \proc{ExtractMin} operations. Previous results were
   either amortized or needed $O(\log n)$ bits of additional state
   information between operations.
\end{abstract}

\section{Introduction}


In 1964 Williams presented ``Algorithm 232'' \cite{W64}, 
commonly known as the binary heap. The binary heap 
is a priority queue data structure storing a dynamic set of $n$ elements
from a totally ordered universe, 
supporting the insertion
of an element (\proc{Insert}) and the deletion of the minimum element
(\proc{ExtractMin}) in worst-case $O(\log n)$ time.
%
The binary heap structure is an \emph{implicit} data structure, i.e.,
it consists of 
an array of length $n$ storing the elements, and
no information is stored  between operations except for the array and
the value~$n$. 
Sometimes data structures storing $O(1)$ additional words are also called
implicit. In this paper we restrict
our attention to \emph{strictly implicit} priority queues, i.e.,
data structures that do not store any additional information
than the array of elements and the value $n$ between operations.

Due to the $\Omega(n\log n)$ lower bound on comparison based sorting,
either \proc{Insert} or \proc{ExtractMin} must take $\Omega(\log
n)$ time, but not necessarily both. 
Carlson \textit{et~al.}~\cite{CarlssonMP88} presented an implicit
priority queue with worst-case $O(1)$ and $O(\log n)$ time
\proc{Insert} and \proc{ExtractMin} operations, respectively.
However, the structure is not strictly implicit since it needs to
 store $O(1)$ additional
words. Harvey and Zatloukal~\cite{postorder} presented a strictly implicit priority structure achieving the same
bounds, but amortized.
No previous strictly implicit priority queue with matching worst-case
time bounds is known.






A measurement often studied in implicit data structures and in-place
algorithms is the number of element moves performed during the execution of a
procedure. 
Francessini showed how to sort $n$ elements implicitly using $O(n\log
n)$ comparisons and $O(n)$ moves \cite{Franceschini07}, and
Franceschini and Munro~\cite{fm06} presented implicit dictionaries
with amortized $O(\log n)$ time updates with amortized $O(1)$ moves
per update.  The latter immediately implies an implicit priority queue
with amortized $O(\log n)$ time \proc{Insert} and \proc{ExtractMin}
operations performing amortized $O(1)$ moves per operation.  No
previous implicit priority queue with $O(1)$ time \proc{Insert} operations
achieving $O(1)$ moves per operation is known.

For a more thorough survey of previous priority queue results,
see~\cite{ianfest13}.

\paragraph{Our Contribution}
We present two strictly implicit priority queues. The first structure
(Section~\ref{sec:amortized_solution}) limits the number of moves to
$O(1)$ per operation with amortized $O(1)$ and $O(\log n)$ time
\proc{Insert} and \proc{ExtractMin} operations, respectively.  However
the bounds are all amortized and it remains an open problem to achieve
these bounds in the worst case for strictly implicit priority queues.
We note
that this structure implies a different way of sorting in-place with
$O(n\log n)$ comparisons and $O(n)$ moves. The second structure
(Section~\ref{sec:worst_case_solution}) improves over
\cite{CarlssonMP88,postorder} by achieving \proc{Insert} and
\proc{ExtractMin} operations with worst-case $O(1)$ and $O(\log n)$
time (and moves), respectively.
The structure in \rsec{sec:worst_case_solution} assumes all elements
to be distinct where as the structure in \rsec{sec:amortized_solution}
also can be extended to support identical elements (see the appendix). See
Figure~\ref{fig:bounds} for a comparison of new and previous results.
\todo{idea of constructions, and discuss role of strictness}


\begin{table}[t]
\def\AM{\llap{\mbox{$\star\;$}}}
\centering
\tabcolsep1ex
\caption{Selected previous and new results for implicit priority queues. The bounds are asymptotic, and ~~\AM are amortized bounds.}
\label{fig:bounds}
\begin{tabular}{lccccc}
&&\proc{Extract-} &&& Identical \\
& \proc{Insert} & \proc{Min} & Moves & Strict & elements \\
\hline
Williams \cite{W64} & $\log n$ & $\log n$ & $\log n$ & yes & yes \\
Carlsson \textit{et~al.}~\cite{CarlssonMP88} & $1$ & $\log n$ & $\log n$ & no & yes \\
Edelkamp \textit{et~al.}~\cite{EdelkampEK13} & $1$ & $\log n$ & $\log n$ & no & yes \\
Harvey and Zatloukal~\cite{postorder} & \AM$1$ & \AM$\log n$ & \AM$\log n$ & yes & yes \\
Franceschini and Munro~\cite{fm06} & \AM$\log n$ & \AM$\log n$ & \AM $1$ & yes &  no\\
Section~\ref{sec:amortized_solution} & \AM$1$ & \AM$\log n$ & \AM$1$ & yes & yes \\
Section~\ref{sec:worst_case_solution} & $1$ & $\log n$ & $\log n$ & yes & no \\
\hline
\end{tabular}
\end{table}

\paragraph{Preliminaries}

We assume the \emph{strictly implicit model} as defined in
\cite{BrodalNT12} where we are only allowed to store the number of
elements $n$ and an array containing the $n$ elements. Comparisons
are the only allowed operations on the elements. The number $n$ is
stored in a memory cell with $\Theta(\log n)$ bits (word size) and any
operation usually found in a RAM is allowed for computations on $n$
and intermediate values. The number of moves is the number of writes
to the array storing the elements. That is, swapping two elements costs two moves.

A fundamental technique in the implicit model is to encode a $0/1$-bit
with a pair of distinct elements $(x,y)$, where the pair encodes
$1$ if $x < y$ and $0$ otherwise.

A binary heap is a complete binary
tree structure where each node stores an element
and the tree satisfies \emph{heap order}, i.e., the element at a
non-root node is larger than or equal to the element at the parent
node.  Binary heaps can be generalized to $d$-ary heaps
\cite{Johnson77}, where the degree of each node is $d$ rather than
two. This implies $O(\log_{d} n)$ and $O(d \log_{d} n)$ time for
\proc{Insert} and \proc{ExtractMin}, respectively, using $O(\log_{d}
n)$ moves for both operations.




%
	
\section{Amortized $O(1)$ moves}
\label{sec:amortized_solution}

In this section we describe a strictly implicit priority queue
supporting amortized $O(1)$ time \proc{Insert} and amortized $O(\log
n)$ time \proc{ExtractMin}. Both operations perform amortized $O(1)$
moves. In Sections
\ref{sec:moverepresentation}-\ref{sec:amortized_analysis} we assume
elements are distinct. In Appendix \ref{appendix:identical} we
describe how to handle identical elements.

\paragraph{Overview}
The basic idea of our priority queue is the following (the details are
presented in Section~\ref{sec:moverepresentation}). The structure
consists of four components: an insertion buffer $B$ of size
$O(\log^3 n)$; $m$ insertion heaps $I_1,I_2,\ldots,I_m$ each of size
$\Theta(\log^3 n)$, where $m=O(n/\log^3 n)$; a singles structure~$T$,
of size $O(n)$; and a binary heap~$Q$, storing $\{1,2,\ldots,m\}$ (integers encoded by pairs of elements) with
the ordering $i\leq j$ if and only if $\min I_i\leq\min I_j$.  Each
$I_i$ and $B$ is a $\log n$-ary heap of size $O(\log^3 n)$.  The
table below summarizes the performance of each component:
\begin{center}
\tabcolsep1ex
  \begin{tabular}{ccccc}
    & \multicolumn{2}{c}{Insert} & \multicolumn{2}{c}{ExtractMin} \\
    Structure     & Time & Moves & Time & Moves \\
    \hline
    $B$, $I_i$    &   1  &   1   & $\log n$ & 1 \\
    $Q$           & $\log^2 n$ & $\log^2 n$ &  $\log^2 n$ & $\log^2 n$ \\
    $T$           & $\log n$ & $1$ & $\log n$ & $1$ \\
    \hline
  \end{tabular}
\end{center}

It should be noted that the implicit dictionary of Franceschini and
Munro \cite{fm06} could be used for $T$, but we will give a more
direct solution since we only need the restricted \proc{ExtractMin}
operation for deletions.

The \proc{Insert} operation inserts new elements into $B$. If the size
of $B$ becomes $\Theta(\log^{3} n)$, then $m$ is incremented by one, $B$
becomes $I_{m}$, 
$m$ is inserted into $Q$,
and $B$ becomes a new empty $\log n$-ary heap.
An \proc{ExtractMin} operation first identifies the minimum element in
$B$, $Q$ and~$T$. If the overall minimum element~$e$ is in $B$ or $T$,
$e$ is removed from $B$ or $T$.
If the minimum element $e$ resided in~$I_i$, where $i$ is stored at the root of $Q$,
then $e$ and $\log^2 n$ further smallest elements are
extracted from $I_i$ (if $I_i$ is not empty) and all except $e$ inserted into $T$ ($T$ has cheap operations whereas $Q$ does not, thus the expensive operation on $Q$ is amortized over inexpensive ones in $T$), and
$i$ is deleted from and reinserted into $Q$ with respect to the new minimum element
in $I_i$.  Finally $e$ is returned.

For the analysis we see that \proc{Insert} takes $O(1)$ time and moves, except
when converting $B$ to a new $I_m$ and inserting $m$ into $Q$. 
The $O(\log^2 n)$ time and moves for this
conversion is amortized over the insertions into~$B$, which
becomes amortized $O(1)$, since $|B|=\Omega(\log^2 n)$. For
\proc{ExtractMin} we observe that an expensive deletion from $Q$ only
happens once for every $\log^2 n$-th element from $I_i$ (the remaining
ones from $I_i$ are moved to $T$ and deleted from $T$), and finally
if there have been $d$ \proc{ExtractMin} operations, then at most $d+m\log^2 n$ elements
have been inserted into~$T$, with a total cost of $O((d+m\log^2 n)\log
n)=O(n+d\log n)$, since $m=O(n/\log^3 n)$.

\subsection{The implicit structure}
\label{sec:moverepresentation}

\fig{h!}{1.0}{width=\textwidth}{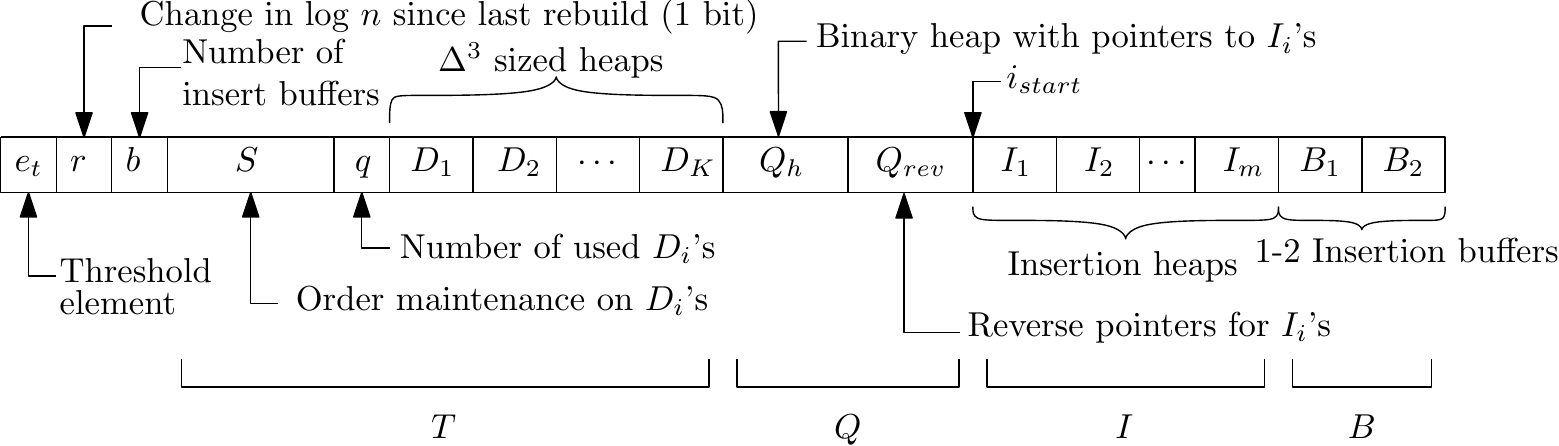}{memory_layout}{The
  different structures and their layout in memory.}

We now give the details of our representation
(see Figure~\ref{fig:memory_layout}). 
We select one element $e_t$ as our \emph{threshold element}, and
denote elements greater than $e_t$ as \emph{dummy elements}.  The
current number of elements in the priority queue is denoted $n$.  We
fix an integer $N$ that is an approximation of $n$, where $N\leq n<4N$
and $N=2^j$ for some $j$.  Instead of storing $N$, we store a bit
$r=\lfloor\log n\rfloor-\log N$, encoded by two dummy elements.  We
can then compute $N$ as $N=2^{\lfloor\log n\rfloor-r}$, where $\lfloor
\log n \rfloor$ is the position of the most significant bit in the
binary representation of $n$ (which we assume is computable in
constant time).  The value $r$ is easily maintained: When $\lfloor
\log n \rfloor$ changes, $r$ changes accordingly.  We let $\Delta =
\log(4N)=\lfloor\log n\rfloor+2-r$, i.e., $\Delta$ bits is sufficient
to store an integer in the range $0..n$. We let $M=\lceil
4N/\Delta^3\rceil$.

We maintain the invariant that the size of the insertion buffer $B$
satisfies $1\leq |B| \leq 2\Delta^3$, and that $B$ is split into two
parts $B_1$ and $B_2$, each being $\Delta$-ary heaps ($B_2$ possibly empty), where
$|B_1|=\min\{|B|,\Delta^3\}$ and
$|B_2|=|B|-|B_1|$. We use two buffers to prevent expensive operation sequences that alternate inserting and deleting the same element. We store a bit $b$ indicating if
$B_2$ is nonempty, i.e., $b=1$ if and only if $|B_2|\neq 0$. The bit $b$ is encoded using two dummy elements.
The structures $I_1,I_2,\ldots,I_m$ are $\Delta$-ary heaps storing
$\Delta^3$ elements. 
The binary heap $Q$ is stored using two arrays $Q_h$ and $Q_{rev}$
each of a fixed size $M\geq m$ and storing integers in the range
$1..m$. Each value in both arrays is encoded using $2\Delta$ dummy
elements, i.e., $Q$ is stored using $4M\Delta$ dummy elements. The
first $m$ entries of $Q_h$ store the binary heap, whereas $Q_{rev}$
acts as reverse pointers, i.e., if $Q_h[j]=i$ then $Q_{rev}[i]=j$.
All operations on a regular binary heap take $O(\log n)$ time, but
since each ``read''/''write'' from/to $Q$ needs to decode/encode an
integer the time increases by a factor $2\Delta$.  It follows that $Q$
supports \proc{Insert} and \proc{ExtractMin} in $O(\log^2 n)$ time,
and \proc{FindMin} in $O(\log n)$ time.

We now describe $T$ and we need the following density
maintenance result.
\vspace{-1mm}
\begin{lemma}[\cite{BrodalFJ02}]
\label{lem:blackbox1}
  There is a dynamic data structure storing $n$ comparable elements in
  an array of length $(1+\varepsilon)n$, supporting \proc{Insert} and
  \proc{ExtractMin} in amortized $O(\log^2 n)$ time and
  \proc{FindPredecessor} in worst case $O(\log n)$ time.
  \proc{FindPredecessor} does not modify the array.
\end{lemma}

\vspace{-3mm}
\begin{corollary}
\label{cor:bb2}
There is an implicit data structure storing $n$ $($key, index$)$
pairs, while supporting \proc{Insert} and \proc{ExtractMin} in
amortized $O(\log^3 n)$ time and moves, and \proc{FindPredecessor} in
$O(\log n)$ time in an array of length $\Delta(2+\varepsilon)n$.
\end{corollary}
\begin{proof}
We use the structure from \rlem{blackbox1} to store pairs of a key and
an index, where the index is encoded using $2\Delta$ dummy elements.
All additional space is filled with dummy elements.  However
comparisons are only made on keys and not indexes, which means we
retain $O(\log n)$ time for \proc{FindMin}. Since the stored elements
are now an $O(\Delta) = \Theta(\log n)$ factor larger, the time for
update operations becomes an $O(\log n)$ factor slower giving
amortized $O(\log^3 n)$ time for \proc{Insert} and
\proc{ExtractMin}. \qed
\end{proof}

The singles structure $T$ intuitively consists of a sorted list of the
elements stored in $T$ partitioned into buckets $D_1,\ldots,D_q$ of
size at most $\Delta^3$, where the minimum element $e$ from bucket
$D_i$ is stored in a structure $S$ from Corollary~\ref{cor:bb2} as the
pair $(e,i)$. Each $D_i$ is stored as a $\Delta$-ary heap of
size $\Delta^3$, where empty slots are filled with dummy
elements. Recall implicit heaps are complete trees, which means
all dummy elements in $D_i$ are stored consecutively after the last
non-dummy element. In $S$ we consider pairs $(e,i)$ where $e>e_t$ to
be empty~spaces.

More specifically, the structure $T$ consists of: $q$, $S$, $D_1, D_2,
\ldots, D_K$, where $K = \lceil\tfrac{N}{16\Delta^3}\rceil \ge q$ is
the number of $D_i$'s available.  The structure $S$ uses $\left\lceil
\frac{N}{4\Delta^2} \right\rceil$ elements and $q$ uses $2\Delta$
elements to encode a pointer. Each $D_i$ uses $\Delta^3$ elements.

The $D_i$'s and $S$ relate as follows. The number of $D_i$'s is at
most the maximum number of items that can be stored in $S$. Let $(e,
i) \in S$, then $\forall x \in D_i : e <
x$, and furthermore for any $(e',i') \in S$ with $e < e'$ we have
$\forall x \in D_i: x < e'$. These invariants do not apply to dummy
elements.
Since $D_i$ is a $\Delta$-ary heap with $\Delta^3$ elements we get
$O(\log_{\Delta} \Delta^3) = O(1)$ time for \proc{Insert} and
$O(\Delta \log_{\Delta} \Delta^3) = O(\Delta)$ for \proc{ExtractMin}
on a $D_i$.
\vspace{-4mm}
\subsection{Operations}

For both \proc{Insert} and \proc{ExtractMin} we need to know $N$,
$\Delta$, and whether there are one or two insert buffers as well as
their sizes. First $r$ is decoded and we compute $\Delta = 2 +
\mathrm{msb}(n) - r$, where $\mathrm{msb}(n)$ is the position of the
most significant bit in the binary representation of $n$ (indexed from
zero). From this we compute $N = 2^{\Delta - 2}$, $K = \lceil N/(16\Delta^3)\rceil$, and $M = \lceil 4N / \Delta^3 \rceil$. By decoding $b$ we
get the number of insert buffers. To find the sizes of $B_1$ and $B_2$
we compute the value $i_{start}$ which is the index of the first element
in $I_1$.
The size of $B_1$ is computed as follows. If $(n - i_{start}) \bmod
\Delta^3 = 0$ then $|B_1| = \Delta^3$. If $B_2$ exists then $B_1$
starts at $n-2\Delta^3$ and otherwise $B_1$ starts at $n -
\Delta^3$. If $B_2$ exists and $(n - i_{start}) \bmod \Delta^3 = 0$
then $|B_2| = \Delta^3$, otherwise $|B_2| = (n-i_{start})
\bmod{\Delta^3}$. Once all of this information is computed the actual
operation can start. If $n = N+1$ and an \proc{ExtractMin} operation
is called, then the \proc{ExtractMin} procedure is executed and
afterwards the structure is rebuilt as described in the paragraph
below. Similarly if $n = 4N-1$ before an \proc{Insert} operation the
new element is appended and the data structure is rebuilt.
\vspace{-2mm}
\paragraph{\proc{Insert}}

If $|B_1| < \Delta^3$ the new element is inserted in $B_1$ by the
standard insertion algorithm for $\Delta$-ary heaps. If $|B_1| =
\Delta^3$ and $|B_2| = 0$ and a new element is inserted the two
elements in $b$ are swapped to indicate that $B_2$ now exists. When
$|B_1| = |B_2| = \Delta^3$ and a new element is inserted, $B_1$
becomes $I_{m+1}$, $B_2$ becomes $B_1$, ${m+1}$ is inserted in $Q$
(possibly requiring $O(\log n)$ values in $Q_h$ and $Q_{rev}$ to be
updated in $O(\log^2 n)$ time).  Finally the new element becomes~$B_2$.

\vspace{-2mm}
\paragraph{\proc{ExtractMin}}

Searches for the minimum element $e$ are performed in $B_1$, $B_2$,
$S$, and $Q$. If $e$ is in $B_1$ or $B_2$ it is deleted, the last
element in the array is swapped with the now empty slot and the usual
bubbling for heaps is performed. If $B_2$ disappears as a result, the
bit $b$ is updated accordingly. If $B_1$ disappears as a result, $I_m$
becomes $B_1$, and $m$ is removed from $Q$.

If $e$ is in $I_i$ then $i$ is deleted from $Q$, $e$ is extracted from
$I_i$, and the last element in the array is inserted in $I_i$. The
$\Delta^2$ smallest elements in $I_i$ are extracted and inserted into
the \emph{singles structure}: for each element a search in $S$ is
performed to find the range it belongs to, i.e.\ $D_j$, the structure
it is to be inserted in. Then it is inserted in $D_j$ (replacing a
dummy element that is put in $I_i$, found by binary search). If $|D_j|
= \Delta^3$ and $q = K$ the priority queue is rebuilt. Otherwise if
$|D_j| = \Delta^3$, $D_j$ is split in two by finding the median $y$ of
$D_j$ using a linear time selection
algorithm~\cite{CarlssonS95}. Elements $\ge y$ in $D_j$ are swapped
with the first $\Delta^3/2$ elements in $D_q$ then $D_j$ and $D_q$
are made into $\Delta$-ary heaps by repeated insertion.  Then $y$ is
extracted from $D_q$ and $(y,q)$ is inserted in $S$.  The dummy
element pushed out of $S$ by $y$ is inserted in~$D_q$.  Finally $q$ is
incremented and we reinsert $i$ into $Q$.  Note that it does not
matter if any of the elements in $I_i$ are dummy elements, the
invariants are still maintained.

If $(e, i) \in S$, the last element of the array is inserted into the
singles structure, which pushes out a dummy element~$z$.  The minimum
element $y$ of $D_i$ is extracted and $z$ inserted instead.  We
replace $e$ by $y$ in $S$.  If $y$ is a dummy element, we update $S$
as if $(y,i)$ was removed.  Finally $e$ is returned. Note this might
make $B_1$ or $B_2$ disappear as a result and the steps above are
executed if needed.

\vspace{-2mm}
\paragraph{Rebuilding}

We let the new $N = n'/2$, where $n'$ is $n$ rounded to the nearest
power of two.  Using a linear time selection
algorithm~\cite{CarlssonS95}, find the element with rank
$n-i_{start}$, this element is the new threshold element $e_t$, and it
is put in the first position of the array. Following $e_t$ are all the
elements greater than $e_t$ and they are followed by all the elements
comparing less than $e_t$. We make sure to have at least $\Delta^3 /
2$ elements in $B_1$ and at most $\Delta^3 / 2$ elements in $B_2$
which dictates whether $b$ encodes $0$ or $1$. The value $q$ is
initialized to $1$. All the $D_i$ structures are considered empty
since they only contain dummy elements. The pointers in $Q_h$ and
$Q_{rev}$ are all reset to the value $0$. All the $I_i$ structures as
well as $B_1$ (and possibly $B_2$) are made into $\Delta$-ary heaps
with the usual heap construction algorithm. For each $I_j$ structure
the $\Delta^2$ smallest elements are inserted in the singles structure
as described in the \proc{ExtractMin} procedure, and $j$ is inserted
into $Q$. The structure now satisfies all the invariants.

\vspace{-4mm}
\subsection{Analysis}
\vspace{-2mm}
\label{sec:amortized_analysis}
In this subsection we give the analysis that leads to the following
theorem.

\begin{theorem}
\label{thm:main}
There is a strictly implicit priority queue supporting \proc{Insert} in amortized $O(1)$ time,
\proc{ExtractMin}  in amortized $O(\log n)$ time. Both operations perform amortized $O(1)$
moves.
\end{theorem}



\paragraph{\proc{Insert}}
While $|B| < 2\Delta^3$, each insertion takes $O(1)$ time.  When an
insertion happens and $|B|=2\Delta^3$, the insertion into $Q$ requires
$O(\log^2 n)$ time and moves. During a sequence of $s$ insertions,
this can at most happen $\lceil s/\Delta^3\rceil$ times, since $|B|$
can only increase for values above $\Delta^3$ by insertions, and each
insertion at most causes $|B|$ to increase by one. The total cost for
$s$ insertions is $O(s+s/\Delta^3\cdot\log^2 n)=O(s)$, i.e., amortized
constant per insertion.


\paragraph{\proc{ExtractMin}}

We first analyze the cost of updating the singles structure. Each
operation on a $D_i$ takes time $O(\Delta)$ and performs $O(1)$
moves. Locating an appropriate bucket using $S$ takes $O(\log n)$ time
and no moves. At least $\Omega(\Delta^3)$ operations must be
performed on a bucket to trigger an expensive bucket split or bucket
elimination in $S$. Since updating $S$ takes $O(\log^3 n)$ time, the
amortized cost for updating $S$ is $O(1)$ moves per insertion and
extraction from the singles structure. In total the operations on the
singles structure require amortized $O(\log n)$ times and amortized
$O(1)$ moves.
For \proc{ExtractMin}
the searches performed all take $O(\log n)$ comparisons and no
moves. If $B_1$ disappears as a result of an extraction we know at
least $\Omega(\Delta^3)$ extractions have occurred because a rebuild
ensures $|B_1| \ge \Delta^3 / 2$. These extractions pay for extracting
$I_m$ from $Q_h$ which takes $O(\log^2 n)$ time and moves, amortized
this gives $O(1/\log n)$ additional time and moves. If the extracted
element was in $I_i$ for some $i$, then $\Delta^2$ insertions occur in
the singles structure each taking $O(\log n)$ time and $O(1)$ moves
amortized. If that happens either $\Omega(\Delta^3)$ insertions or
$\Delta^2$ extractions have occurred: Suppose no elements from $I_i$
have been inserted in the singles structure, then the reason there is
a pointer to $I_i$ in $Q_h$ is due to $\Omega(\Delta^3)$
insertions. When inserting elements in the singles structure from
$I_i$ the number of elements inserted is $\Delta^2$ and these must
first be deleted. From this discussion it is evident that we have
saved up $\Omega(\Delta^2)$ moves and $\Omega(\Delta^3)$ time, which
pay for the expensive extraction. Finally if the minimum element was
in $S$, then an extraction on a $\Delta$-ary heap is performed which
takes $O(\Delta)$ time and $O(1)$ moves, since its height is $O(1)$.

\paragraph{Rebuilding}
The cost of rebuilding is $O(n)$, due to a selection and building
heaps with $O(1)$ height. There are three reasons a rebuild might
occur:
\begin{inparaenum}[(i)]
  \item $n$ became $4N$,
  \item $n$ became $N-1$, or
  \item An insertion into $T$ would cause $q > K$.
\end{inparaenum}
By the choice of $N$ during a rebuild it is guaranteed that in the
first and second case at least $\Omega(N)$ insertions or extractions
occurred since the last rebuild, and we have thus saved up at least
$\Omega(N)$ time and moves. For the last case we know that each
extraction incur $O(1)$ insertions in the singles structure in an
amortized sense. Since the singles structure accommodates $\Omega(N)$
elements and a rebuild ensures the singles structure has $o(n)$ non
dummy elements (\rlem{limit_regular}), at least $\Omega(N)$
extractions have occurred which pay for the rebuild.
\vspace{-1mm}
\begin{lemma}
\label{lem:limit_regular}
Immediately after a rebuild $o(n)$ elements in the singles structure
are non-dummy elements
\end{lemma}
\vspace{-3mm}
\begin{proof}
There are at most $n/\Delta^3$ of the $I_i$ structures and $\Delta^2$
elements are inserted in the singles structure from each $I_i$, thus
at most $n/\Delta = o(n)$ non-dummy elements reside in the singles
structure after a rebuild. \qed
\end{proof}
The paragraphs above establish \rthm{main}.
\vspace{-3mm}

\section{Worst case solution}
\label{sec:worst_case_solution}

In this section we present a strictly implicit priority queue
supporting \proc{Insert} in worst-case $O(1)$ time and
\proc{ExtractMin} in worst-case $O(\log n)$ time (and moves). The data
structure requires all elements to be distinct.
The main concept used is a variation on binomial trees. The priority queue
is a forest of $O(\log n)$ such trees. We start with a discussion of
the variant we call \emph{relaxed binomial trees}, then we describe
how to maintain a forest of these trees in an amortized sense, and
finally we give the deamortization.
\vspace{-4mm}
\subsection{Relaxed binomial tree}

\emph{Binomial trees} are defined inductively: A single node is a
binomial tree of size one and the node is also the root. A binomial
tree of size $2^{i+1}$ is made by \emph{linking} two binomial trees
$T_1$ and $T_2$ both of size $2^i$, such that one root becomes the
rightmost child of the other root.
We lay out in memory 
a binomial
tree of size $2^{i}$ 
by a preorder traversal of
the tree where children are visited in order of increasing size,
i.e.\ $c_0, c_1, \ldots, c_{i-1}$. This layout is also described in
\cite{CarlssonMP88}. See \rfig{rbq} for an illustration of the
layout. 
In a \emph{relaxed binomial tree} (RBT) each nodes stores an element, 
satisfying the following order:
Let $p$ be a node with $i$ children, and let $c_j$ be a child of $p$.
Let $T_{c_j}$ denote the set of elements in the subtree rooted at
$c_j$. We have the invariant that the element $c_\ell$ is less than
either all elements in $T_{c_\ell}$ or less than all elements in
$\bigcup_{j < \ell} T_{c_j}$ (see \rfig{rbq}). In particular we have
the requirement that the root must store the smallest element in the
tree. In each node we store a flag indicating in which direction the
ordering is satisfied.  Note that linking two adjacent RBTs of equal
size can be done in $O(1)$ time: compare the keys of the two roots, if
the lesser is to the right, swap the two nodes and finally update the
flags to reflect the changes as just described. \todo{details?}

For an unrelated technical purpose we also
need to store whether a node is the root of a RBT. This
information is encoded using three elements per node (allowing $3! = 6$
permutations, and we only need to differentiate between three states per node: ``root'', ``minimum
of its own subtree'', or ``minimum among strictly smaller subtrees'').

\fig{h}{1.0}{}{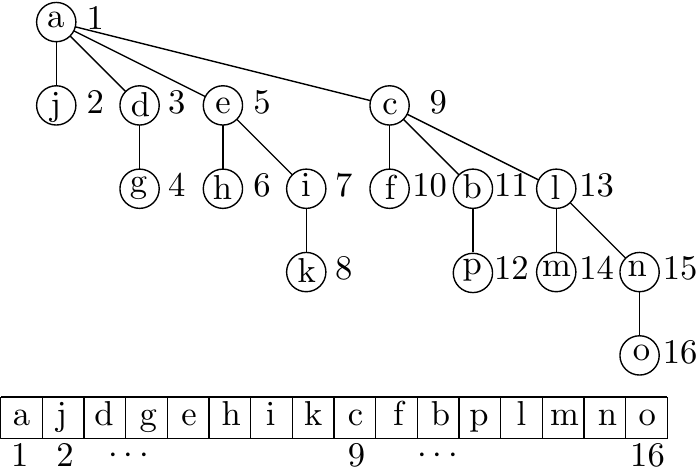}{rbq}{An example of an RBT on 16 elements
  (a,b,...,o). The layout in memory of an RBT and a regular
  binomial tree is the same. Note here that node 9 has element $c$ and is
  not the minimum of its subtree because node 11 has element $b$, but $c$
  is the minimum among the subtrees rooted at nodes $2$, $3$, and $5$
  ($c_0$, $c_1$, and $c_2$). Note also that node $5$ is the minimum of
  its subtree but not the minimum among the trees rooted at nodes 2
  and 3, which means only one state is valid. Finally node $3$ is the
  minimum of both its own subtree and the subtree rooted at node $2$,
  which means both states are valid for that node.}

To extract the minimum element of an RBT it is replaced by another
element. The reason for replacing is that the forest of RBTs is
implicitly maintained in an array and elements are removed from the
right end, meaning only an element from the last RBT is removed. If
the last RBT is of size $1$, it is trivial to remove the
element. If it is larger, then we \emph{decompose} it. We
first describe how to perform a \proc{Decompose} operation which
changes an RBT of size $2^i$ into $i$ structures
$T_{i-1},\ldots,T_1,T_0$, where $|T_j| = 2^j$.
Then we describe how to perform
\proc{ReplaceMin} which takes one argument, a new element, and
extracts the minimum element from an RBT and inserts the argument in
the same structure.

A \proc{Decompose} procedure is essentially reversing insertions. We
describe a tail recursive procedure taking as argument a node~$r$.  If
the structure is of size one, we are done. If the structure is of size
$2^i$ the $(i-1)$th child, $c_{i-1}$, of $r$ is inspected, if it is
not the minimum of its own subtree, the element of $c_{i-1}$ and $r$
are swapped. The $(i-1)$th child should now encode ``root'', that way
we have two trees of size $2^{i-1}$ and we recurse on the subtree to
the right in the memory layout. This procedure terminates in $O(i)$
steps and gives $i+1$ structures of sizes $2^{i-1},
2^{i-2},\ldots,2,1$, and $1$ laid out in decreasing order of size
(note there are two structures of size $1$). This enables easy removal
of a single element.

The \proc{ReplaceMin} operation works similarly to the
\proc{Decompose}, where instead of always recursing on the right, we
recurse where the minimum element is the root. When the recursion
ends, the minimum element is now in a structure of size $1$, which is
deleted and replaced by the new element. The decomposition is then
reversed by linking the RBTs using the \proc{Link} procedure. Note it
is possible to keep track of which side was recursed on at every level
with $O(\log n)$ extra bits, i.e.\ $O(1)$ words. The operation takes
$O(\log n)$ steps and correctness follows by the \proc{Decompose} and
\proc{Link} procedures. This concludes the description of RBTs and
yields the following theorem.

\begin{theorem}
  \label{thm:rbq}
  On an RBT with $3 \cdot 2^i$ elements, \proc{Link} and \proc{FindMin}
  can be supported in $O(1)$ time and \proc{Decompose} and
  \proc{ReplaceMin} in $O(i)$ time.
\end{theorem}
\vspace{-6mm}
\subsection{How to maintain a forest}

As mentioned our priority queue is a forest of the relaxed binomial
trees from \rthm{rbq}. An easy amortized solution is to store one
structure of size $3 \cdot 2^j$ for every set bit $j$ in the binary
representation of $\lfloor n/3 \rfloor$. During an insertion this
could cause $O(\log n)$ \proc{Link} operations, but by a similar
argument to that of binary counting, this yields $O(1)$ amortized
insertion time. We are aiming for a worst case constant time solution
so we maintain the invariant that there are at most $5$ structures of
size $2^i$ for $i=0,1,\ldots,\lfloor \log n \rfloor$. This enables us
to postpone some of the \proc{Link} operations to appropriate
times. We are storing $O(\log n)$ RBTs, but we do not store which
sizes we have, this information must be decodable in constant time
since we do not allow storing additional words. Recall that we need
$3$ elements per node in an RBT, thus in the following we let $n$ be
the number of elements and $N = \lfloor n/3 \rfloor$ be the number of
nodes. We say a node is in node position $k$ if the three elements in
it are in positions $3k-2$, $3k-1$, and $3k$. This means there is a
buffer of $0,1$, or $2$ elements at the end of the array. When a third
element is inserted, the elements in the buffer become an RBT with a
single node and the buffer is now empty. If an \proc{Insert} operation
does not create a new node, the new element is simply appended to the
buffer. We are not storing the structure of the forest (i.e.\ how many
RBTs of size $2^j$ exists for each $j$), since that would require
additional space. To be able to navigate the forest we need the
following two lemmas.
\vspace{-1mm}
\begin{lemma}
  \label{lem:iffpositions}
  There is a structure of size $2^i$ at node positions
  $k,k+1,\ldots,k+2^i-1$ if and only if the node at position $k$ encodes
  ``root'', the node at position $k+2^i$ encodes ``root'' and the node
  at position $k+2^{i-1}$ encodes ``not root''.
\end{lemma}
\vspace{-3mm}
\begin{proof}
  It is trivially true that the mentioned nodes encode ``root'',
  ``root'' and ``not root'' if an RBT with $2^i$ nodes is present in
  those locations.

  We first observe there cannot be a structure of size $2^{i-1}$
  starting at position~$k$, since that would force the node at
  position $k+2^{i-1}$ to encode ``root''. Also all structures between
  $k$ and $N$ must have less than $2^i$ elements, since both nodes at
  positions $k$ and $k+2^i$ encode ``root''. We now break the analysis
  in a few cases and the lemma follows from a proof by
  contradiction. Suppose there is a structure of size $2^{i-2}$
  starting at $k$, then for the same reason as before there cannot be
  another one of size $2^{i-2}$. Similarly, there can at most be one
  structure of size $2^{i-3}$ following that structure. Now we can
  bound the total number of nodes from position $k$ onwards in the
  structure as: $2^{i-2} + 2^{i-3} + 5\sum_{j=0}^{i-4} 2^j = 2^i - 5 <
  2^i$, which is a contradiction. So there cannot be a structure of
  size $2^{i-2}$ starting at position $k$. Note there can at most be
  three structures of size $2^{i-3}$ starting at position $k$, and we
  can again bound the total number of nodes as: $3 \cdot 2^{i-3} +
  5\sum_{j=0}^{i-4} 2^j = 2^i - 5 < 2^i$, again a contradiction. \qed
\end{proof}
\vspace{-3mm}
\begin{lemma}
  \label{lem:positions}
  If there is an RBT with $2^i$ nodes the root is in position $N -
  2^i k - x + 1$ for $k=1,2,3,4$ or $5$ and $x = N \bmod 2^i$.
\end{lemma}
\vspace{-3mm}
\begin{proof}
  There are at most $5 \cdot 2^i - 5$ nodes in structures of size $\le
  2^{i-1}$. All structures of size $\ge 2^i$ contribute $0$ to $x$,
  thus the number of nodes in structures with $\le 2^{i-1}$ nodes must
  be $x$ counting modulo $2^i$. This gives exactly the five
  possibilites for where the first tree of size $2^{i}$ can be. \qed
\end{proof}

We now describe how to perform an \proc{ExtractMin}. First, if there
is no buffer ($n \bmod 3 = 0$) then \proc{Decompose} is executed on
the smallest structure. We apply \rlem{positions} iteratively for
$i=0$ to $\lfloor \log N \rfloor$ and use \rlem{iffpositions} to find
structures of size $2^i$. If there is a structure we call the
\proc{FindMin} procedure (i.e.\ inspect the element of the root node)
and remember which structure the minimum element resides in. If the
minimum element is in the buffer, it is deleted and the rightmost
element is put in the empty position. If there is no buffer, we are
guaranteed due to the first step that there is a structure with $1$
node, which is now the buffer. On the structure with the minimum
element \proc{ReplaceMin} is called with the rightmost element of the
array. The running time is $O(\log n)$ for finding all the structures,
$O(\log n)$ for decomposing the smallest structure and $O(\log n)$ for
the \proc{ReplaceMin} procedure, in total we get $O(\log n)$ for
\proc{ExtractMin}.

The \proc{Insert} procedure is simpler but the correctness proof is
somewhat involved. A new element is inserted in the buffer, if the
buffer becomes a node, then the \emph{least significant bit} $i$ of
$N$ is computed. If at least two structures of size $2^i$ exist (found
using the two lemmas above), then they are linked and become one
structure of size $2^{i+1}$.
\vspace{-1mm}
\begin{lemma}
The \proc{Insert} and \proc{ExtractMin} procedures maintain that at
most five structures of size $2^i$ exist for all $i \le \lfloor \log n
\rfloor$.
\end{lemma}
\begin{proof}
Let $N_{\le i}$ be the total number of nodes in structures of size
$\le 2^i$. Then the following is an invariant for $i = 0,1,\ldots,
\lfloor \log N \rfloor$.

\begin{equation*}
  N_{\le i} + (2^{i+1} - ((N + 2^i) \bmod 2^{i+1}))) \le 6\cdot 2^i - 1
\end{equation*}

The invariant states that $N_{\le i}$ plus the number of inserts
until we try to link two trees of size $2^i$ is at most $6 \cdot 2^i
- 1$. Suppose that a new node is inserted and $i$ is not the least
significant bit of $N$ then $N_{\le i}$ increases by one and so does
$(N + 2^i) \bmod{2^{i+1}}$, which means the invariant is
maintained. Suppose that $i$ is the least significant bit in $N$
(i.e. we try to link structures of size $2^i$) and there are at least
two structures of size $2^i$, then the insertion makes $N_{\le i}$
decrease by $2 \cdot 2^i - 1 = 2^{i+1} - 1$ and $2^{i+1} - (N + 2^i
\Mod 2^{i+1}))$ increases by $2^{i+1} - 1$, since $(N + 2^i) \bmod
2^{i+1}$ becomes zero, which means the invariant is maintained. Now
suppose there is at most one structure of size $2^i$ and $i$ is the
least significant bit of $N$. We know by the invariant that $N_{\le
  i-1} + (2^{i} - (N + 2^{i-1} \bmod 2^{i})) \le 6 \cdot 2^{i-1}-1$
which implies $N_{\le i-1} \le 6 \cdot 2^{i-1} - 1 - 2^i + 2^{i-1} =
5\cdot2^{i-1} -1$. Since we assumed there is at most one structure of
size $2^{i}$ we get that $N_{\le i} \le 2^i + N_{\le i -1} \le 2^i +
5\cdot 2^{i-1} - 1 = 3.5 \cdot 2^i - 1$. Since $N \bmod{2^{i+1}} =
2^i$ ($i$ is the least significant bit of $N$) we have $N_{\le i} +
(2^{i+1} - (N +2^i \Mod 2^{i+1})) \le 3.5\cdot 2^{i} - 1 + 2^{i+1} =
5.5 \cdot 2^i - 1 < 6\cdot 2^i - 1$.

The invariant is also maintained when deleting: for each $i$ where
$N_i>0$ before the \proc{ExtractMin}, $N_i$ decreases by one. For all
$i$ the second term increases by at most one, and possibly decreases
by $2^{i+1}-1$. Thus the invariant is maintained for all $i$ where
$N_i>0$ before the procedure. If $N_i = 0$ before an
\proc{ExtractMin}, we get $N_j = 2^{j+1} - 1$ for $j \le i$. Since the
second term can at most contribute $2^{j+1}$, we get $N_j
+ (2^{j+1} - ((N + 2^j) \bmod 2^{j+1})) \le 2^{j+1} - 1 + 2^{j+1} \le
6 \cdot 2^{j} - 1$, thus the invariant is maintained. \qed
\end{proof}

Correctness and running times of the procedures have now been established.
\vspace{-5mm}

\bibliographystyle{abbrv}
\bibliography{bib}
\newpage
\appendix
\section{Handling identical elements in the amortized case}
\label{appendix:identical}
The primary difficulty in handling identical elements is that we lose
the ability to encode bits. The primary goal of this section is to do
so anyway. The idea is to let the items stored in the priority queue
be pairs of distinct elements where the key of an item is the lesser
element in the pair. In the case where it is not possible to make a
sufficient number of pairs of distinct elements, almost all elements are equal
and this is an easy case to handle. Note that many pairs (or all for
that matter) can contain the same elements, but each pair can now
encode a bit, which is sufficient for our purposes.

The structure is almost the same as before, however we put a few more
things in the picture. As mentioned we need to use \emph{pairs of
distinct} elements, so we create a mechanism to produce
these. Furthermore we need to do some book keeping such as storing a
pointer and being able to compute whether there are enough pairs of
distinct elements to actually have a meaningful structure. The changes
to the memory layout is illustrated in \rfig{memory_layout_identical}.

\fig{h!}{1.0}{width=\textwidth}{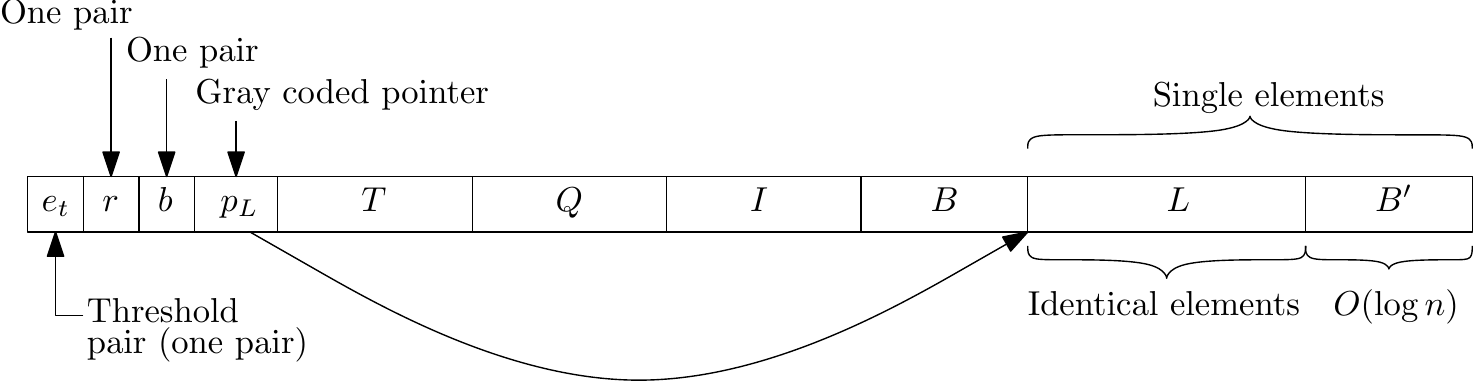}{memory_layout_identical}{The
  different structures and their layout in memory.}

\paragraph{Modifications}
The areas $L$ and $B'$ in memory are used to produce pairs of distinct
elements. The area $p_L$ is a Gray coded pointer\footnote{Gray, F.:
Pulse code communications. U.S. Patent (2632058) (1953)}\cite{gray} with
$\Theta(\log n)$ pairs, pointing to the beginning of $L$. The rest of
the structure is essentially the same as before, except instead of
storing elements, we now store pairs $e = (e_1, e_2)$ and the key of
the pair is $e_k = \min\{e_1,e_2\}$. All comparisons between items are
thus made with the key of the pair. We will refer to the priority
queue from \rsec{sec:amortized_solution} as PQ.

There are a few minor modifications to PQ. Recall that we needed to
simulate \emph{empty spaces} inside $T$ (specifically in $S$,
see \rfig{memory_layout}). The way we simulated empty spaces was by
having elements that compared greater than $e_t$. Now $e_t$ is
actually a pair, where the minimum element is the threshold
element. It might be the case that there are many items comparing
equal to $e_t$, which means some would be used to simulate empty
spaces and others would be actual elements in PQ and some would be
used to encode pointers. This means we need to be able to
differentiate these types that might all compare equal to $e_t$. First
observe that items used for pointers are always located in positions
that are distinguishable from items placed in positions used
as \emph{actual} items. Thus we do not need to worry about confusing
those two. Similarly, the ``empty'' spaces in $T$ are also located in
positions that are distinguishable from pointers. Now we only need to
be able to differentiate ``empty'' spaces and occupied spaces where
the keys both compare equal to $e_t$. Letting items (i.e. pairs) used
as empty spaces encode $1$, and the ``occupied'' spaces encode $0$,
empty spaces and occupied spaces become differentiable as
well. Encoding that bit is possible, since they are not used for
encoding anything else.

Since many elements could now be identical we need to decide whether
there are enough distinct elements to have a meaningful structure. As
an invariant we have that if the two elements in the pair $e_t =
(e_{t,1}, e_{t,2})$ are equal then there are not enough elements to
make $\Omega(\log n)$ pairs of distinct elements. The $O(\log n)$
elements that are different from the majority are then stored at the
end of the array. After every $\log n$th insertion it is easy to check
if there are now sufficient elements to make $\ge c \log n$ pairs for
some appropriately large and fixed $c$. When that happens, the
structure in \rfig{memory_layout_identical} is formed, and $e_t$ must
now contain two distinct elements, with the lesser being the threshold
key. Note also, that while $e_{t,1} = e_{t,2}$ an \proc{ExtractMin}
procedure simply needs to scan the last $< c\log n$ elements and
possibly make one swap to return the minimum and fill the empty index.

\paragraph{Insert}

The structure $B'$ is a list of single elements which functions as an
insertion buffer, that is elements are simply appended to $B'$ when
inserted. Whenever $n \bmod \log n = 0$ a procedure making pairs is
run: At this point we have time to decode $p_L$, and up to $O(\log n)$
new pairs can be made using $L$ and $B'$. To make pairs $B'$ is read,
all elements in $B'$ that are equal to elements in $L$, are put after
$L$, the rest of the elements in $B'$ are used to create pairs using
one element from $L$ and one element from $B'$. If there are more
elements in $B'$, they can be used to make pairs on their own. These
pairs are then inserted into PQ. To make room for the newly inserted
pairs, $L$ might have to move right and we might have to update
$p_L$. Since $p_L$ is a Gray coded pointer, we only need as many bit
changes as there are pairs inserted in PQ, ensuring $O(1)$ amortized
moves. Note that the size of PQ is now the value of $p_L$, which means
all computations involving $n$ for PQ should use $p_L$ instead.

\paragraph{ExtractMin}

To extract the minimum a search for the minimum is performed in PQ,
$B'$ and $L$. If the minimum is in PQ, it is extracted and the other
element in the pair is put at the end of $B'$. Now there are two empty
positions before $L$, so the last two elements of $L$ are put there,
and the last two elements of $B'$ are put in those positions. Note
$p_L$ also needs to be decremented. If the minimum is in $B'$, it is
swapped with the element at position~$n$, and returned. If the minimum
is in $L$, the last element of $L$ is swapped with the element at
position $n$, and it is returned.

\paragraph{Analysis}

Firstly observe that if we can prove the producing of pairs uses
amortized $O(1)$ moves for \proc{Insert} and \proc{ExtractMin} and
$O(1)$ and $O(\log n)$ time respectively, then the rest of the
analysis from \rsec{sec:amortized_analysis} carries through. We first
analyze \proc{Insert} and then \proc{ExtractMin}.

For \proc{Insert} there are two variations: either append elements to
$B'$ or clean up $B'$ and insert into PQ. Cleaning up $B'$ and
inserting into PQ is expensive and we amortize it over the cheap
operations. Each operation that just appends to $B'$ costs $O(1)$ time
and moves. Cleaning up $B'$ requires decoding $p_L$, scanning $B'$ and
inserting $O(\log n)$ elements in PQ. Note that between two clean-ups
either $O(\log n)$ elements have been inserted or there has been at
least one \proc{ExtractMin}, so we charge the time there. Since each
insertion into PQ takes $O(1)$ time and moves amortized we get the
same bound when performing those insertions. The cost of reading $p_L$
is $O(\log n)$, but since we are guaranteed that either $\Omega(\log
n)$ insertions have occurred or at least one \proc{ExtractMin}
operation we can amortize the reading time.

\end{document}